\documentclass[journal]{IEEEtran}
\IEEEoverridecommandlockouts
\usepackage{algorithm}
\usepackage{array}
\usepackage[caption=false,font=normalsize,labelfont=sf,textfont=sf]{subfig}
\usepackage{tabularx,booktabs,makecell}
\usepackage{textcomp}
\usepackage{stfloats}
\usepackage{url}
\usepackage{verbatim}
\usepackage{graphicx}
\usepackage{cite}
\usepackage{amsmath,amssymb,amsfonts,amsthm,mathtools}
\usepackage{xcolor}
\usepackage{blindtext}
\usepackage{lipsum}
\usepackage{subfig}
\usepackage{cuted}
\usepackage{physics}
\usepackage[long]{optidef}
\usepackage{multirow}
\usepackage{steinmetz}
\usepackage{bigints}
\usepackage{algpseudocode}
\usepackage{graphicx}
\usepackage{xcolor}
\usepackage{placeins}

\usepackage{tikz}
\usetikzlibrary{matrix, positioning, shapes, backgrounds, fit, arrows.meta}

\definecolor{chaosblue}{RGB}{30, 100, 200}
\definecolor{chaosgreen}{RGB}{0, 150, 100}
\definecolor{chaosred}{RGB}{200, 50, 50}
\definecolor{chaospurple}{RGB}{150, 50, 200}
\definecolor{chaosorange}{RGB}{230, 120, 0}

\newcommand{\R}{\mathbb{R}}

\newtheorem{definition}{Definition}[section]

\newtheorem{lemma}[definition]{Lemma}
\newtheorem{proposition}[definition]{Proposition}

\newcommand{\E}{\mathbb{E}}
\newcommand{\Var}{\mathrm{Var}}
\newcommand{\ind}{\mathbb{I}}

\begin{document}
\title{Wavelet-Packet-based Noise Signatures With Higher-Order Statistics for Anomaly Prediction}

\author{Indrakshi Dey, Ilias Cherkaoui and Mohamed Khalafalla Hassan\\
Walton Institute, South East Technological University, Waterford, Ireland%
}

\maketitle







\begin{abstract}
This note develops the first-ever noise-centric anomaly prediction method for a fused discrete-time signal. A Wavelet Packet Transform (WPT) provides a time--frequency expansion in which structure and residual can be separated via orthogonal projection. Higher-Order Statistics (HOS), particularly the third-order cumulant (and its bispectral interpretation), quantify non-Gaussianity and nonlinear coupling in the extracted residual. Compact noise signatures are constructed and an analytically calibrated Mahalanobis detector yields a closed-form decision rule with noncentral chi-square performance under mean-shift alternatives. Propositions and proofs establish orthonormality, energy preservation, Gaussian-null behavior of cumulants, and the resulting test statistics. Algorithmic steps and complexity are summarized.
\end{abstract}

\section{Introduction}

In fused wireless sensing and monitoring pipelines (e.g., multi-sensor, multi-antenna, or cross-modal fusion), conventional anomaly detectors often emphasize changes in signal amplitude, energy, or second-order statistics (autocorrelation, power spectrum). However, many subtle anomalies primarily manifest as changes in the \emph{residual micro-structure}---for example, impulsive disturbances, asymmetric fluctuations, or weak nonlinear couplings arising from new scattering objects, device faults, or environmental regime shifts. Such effects can be difficult to detect using only second-order descriptors because second-order statistics are blind to certain forms of non-Gaussianity and higher-order dependence \cite{NikiasPetropulu,Hinich,MendelHOS}.

A complementary viewpoint is to treat ``noise'' not as a nuisance to be removed, but as a \emph{signal-bearing channel} whose statistical structure encodes system state. This motivates a \emph{noise signature} approach: (i) extract residual components using a time--frequency representation with strong localization, and (ii) characterize residual structure using higher-order statistics that vanish under nominal Gaussian assumptions. The present note adopts the \emph{Wavelet Packet Transform (WPT)} for extraction because WPT generalizes the discrete wavelet transform by recursively splitting both approximation and detail subspaces, thereby enabling fine-grained frequency partitioning beyond standard dyadic wavelets \cite{CoifmanWickerhauser,WickerhauserBook,Mallat}. For statistical characterization, we use \emph{Higher-Order Statistics (HOS)}---specifically third-order cumulants and their frequency-domain counterpart, the bispectrum---because they reveal nonlinear interactions and phase coupling and are identically zero for Gaussian processes \cite{NikiasPetropulu,Hinich,MendelHOS}.

Finally, to obtain an analytically calibrated anomaly detector, we model the nominal signature vector as Gaussian and use the \emph{Mahalanobis distance} \cite{Mahalanobis1936,Mardia} with a chi-square threshold, yielding closed-form false-alarm calibration and noncentral chi-square detection performance under mean shift \cite{KayDetection}. For persistence-aware early warning, we optionally integrate \emph{CUmulative SUM (CUSUM)} sequential testing \cite{PageCUSUM}. Surveys of anomaly detection highlight the practical importance of feature design and calibration under nonstationarity \cite{ChandolaSurvey}; our approach emphasizes physically interpretable time--frequency features plus HOS features whose ``null'' behavior is theoretically grounded.

\section{Physical Signal Model}

Let $x[n]\in\R$ denote the fused discrete-time data sequence. We process $x[n]$ in frames of length $N$ samples:
\begin{equation}
\mathbf{x}_m \triangleq \{x[mN+n]\}_{n=0}^{N-1}, \qquad m=0,1,\dots
\end{equation}
Within one frame (index $m$ omitted),
\begin{equation}
x[n] = s[n] + v[n], \qquad n=0,\dots,N-1. \label{eq:additive_model}
\end{equation}
where, $x[n]$ is the fused observable, $s[n]$ is the structured component (coherent waveform content, stable propagation patterns, slowly varying baseline), $v[n]$ is the residual (sensor electronics, channel micro-fluctuations, interference remnants, environmental micro-dynamics). The expectation operator is $\E[\cdot]$ and the variance is $\Var(\cdot)$.

\section{Wavelet Packet Residual Extraction}

Subband decompositions represent a signal as a sum of components localized in \emph{frequency} (via bandpass filtering) and \emph{time} (via shifting localized basis functions). Such representations are attractive when a structured component concentrates in a small subset of time--frequency cells, while residual components distribute more diffusely \cite{Mallat,VetterliKovacevic}. Standard discrete wavelet transforms recursively split only the low-frequency (approximation) branch, producing a fixed dyadic frequency partition. Wavelet packets generalize this by recursively splitting \emph{both} approximation and detail branches, yielding richer orthonormal bases and finer frequency tiling \cite{CoifmanWickerhauser,WickerhauserBook}. This is useful when residual patterns lie in intermediate frequency bands.

Let $\{\psi_{j,k,u}[n]\}$ denote an \emph{orthonormal wavelet packet basis} indexed by scale $j\in\{0,1,\dots,J\}$, node $k\in\{0,1,\dots,2^j-1\}$, and translation $u\in\mathbb{Z}$. The WPT coefficient is
\begin{equation}
w_{j,k}[u] \triangleq \ip{x}{\psi_{j,k,u}}
= \sum_{n=0}^{N-1} x[n]\psi_{j,k,u}[n]. \label{eq:wpt_inner_product}
\end{equation}
Here $(j,k)$ identifies a frequency subband and $u$ identifies time localization. Hence $w_{j,k}[u]$ quantifies the contribution of $x[n]$ within a specific time--frequency region.

An efficient orthonormal implementation of the Wavelet Packet Transform (WPT) is naturally expressed using a two-channel subband filterbank. In this construction, the analysis stage is specified by a \emph{Quadrature Mirror Filter (QMF)} pair consisting of a discrete-time low-pass filter $h[\cdot]$ and a high-pass filter $g[\cdot]$ \cite{VetterliKovacevic}. The term ``mirror'' reflects the complementary spectral roles of the two filters: $h[\cdot]$ extracts low-frequency content while $g[\cdot]$ extracts high-frequency content, and together they enable a lossless (invertible) analysis/synthesis pipeline under suitable design conditions. The key structural property enabling orthonormal wavelet packet atoms is \emph{paraunitarity}. A filterbank is paraunitary when its polyphase matrix is unitary on the unit circle; equivalently, the analysis operator preserves inner products and energy (up to boundary handling), meaning that it is an isometry from the signal space to the coefficient space \cite{VetterliKovacevic,Mallat}. This property is not merely technical: it implies that the wavelet packet functions generated by recursively applying the split inherit orthonormality, so coefficient selection and reconstruction become true orthogonal projections in $\ell_2$, a fact that later supports optimality statements for residual extraction.

The WPT filterbank is realized by repeated filtering and downsampling along a binary decomposition tree. We initialize the root sequence as $w_{0,0}[n]=x[n]$, where $x[n]\in\R$ is the fused frame under analysis. At each level $j$, every node $(j,k)$ is split into two children that represent the low-pass and high-pass subband components of that node. Specifically, the recursion is
\begin{align}
w_{j+1,2k}[u]   &= \sum_{n} h[n-2u]\, w_{j,k}[n], \nonumber\\
w_{j+1,2k+1}[u] &= \sum_{n} g[n-2u]\, w_{j,k}[n], \label{eq:wpt_recursion}
\end{align}
where $j\in\{0,1,\dots,J-1\}$ is the depth (scale), $k\in\{0,1,\dots,2^j-1\}$ indexes the node (subband) at level $j$, and $u$ indexes the downsampled time coordinate within that node. Physically, \eqref{eq:wpt_recursion} can be interpreted as routing the signal through a bank of increasingly narrow bandpass channels: $(j,k)$ selects a frequency interval whose width decreases with $j$, while $u$ tracks the localized content over time. Because wavelet packets recursively split \emph{both} approximation and detail branches, the resulting time--frequency tiling is more flexible than the standard discrete wavelet transform, which splits only the approximation branch; this additional flexibility is useful when residual or anomaly signatures concentrate in intermediate frequency bands that would otherwise remain mixed \cite{CoifmanWickerhauser,WickerhauserBook,Mallat}.

The mathematical backbone of the extraction scheme is the orthonormality of the induced wavelet packet atoms and the resulting energy preservation. Let $\{\psi_{j,k,u}\}$ denote the wavelet packet atoms (basis functions) indexed by depth $j$, node $k$, and translation $u$, so that the WPT coefficients are inner products $w_{j,k}[u]=\langle x,\psi_{j,k,u}\rangle$. Under orthonormal QMF design, these atoms are orthonormal in $\ell_2$, which is stated precisely in the following lemma.

\begin{lemma}[Orthonormality of wavelet packet atoms]
Assume $(h,g)$ form an orthonormal QMF pair generating an orthonormal wavelet packet library. Then
\begin{equation}
\ip{\psi_{j,k,u}}{\psi_{j',k',u'}}=
\delta_{j,j'}\delta_{k,k'}\delta_{u,u'}. \label{eq:orthonormality}
\end{equation}
\end{lemma}
\begin{proof}
Orthonormal QMFs yield a paraunitary two-channel analysis operator \cite{VetterliKovacevic}. Paraunitarity implies that the analysis operator preserves inner products, i.e., it is an isometry on the underlying subspace at each split. Each splitting step therefore maps an orthonormal basis of the parent subspace into orthonormal bases of two mutually orthogonal child subspaces. Recursively applying the paraunitary split along the packet tree preserves orthonormality across all generated atoms, yielding \eqref{eq:orthonormality} \cite{CoifmanWickerhauser,WickerhauserBook}.
\end{proof}

Orthonormality immediately implies a Parseval-type conservation law for coefficient energy, which is crucial because it makes coefficient partitioning physically interpretable and mathematically well-posed.

\begin{proposition}[Parseval identity for WPT]
For any frame $x\in\R^N$,
\begin{equation}
\norm{x}_2^2 = \sum_{k=0}^{2^J-1}\sum_u \left|w_{J,k}[u]\right|^2. \label{eq:parseval}
\end{equation}
\end{proposition}
\begin{proof}
By Lemma~1, $\{\psi_{J,k,u}\}$ is an orthonormal basis for the signal space (up to the adopted boundary convention). Expanding $x$ in this orthonormal basis and applying Parseval's identity yields $\norm{x}_2^2=\sum_{k,u}|\langle x,\psi_{J,k,u}\rangle|^2$, which is exactly \eqref{eq:parseval} \cite{Mallat}.
\end{proof}

A conceptually important consequence of Lemma~1 and Proposition~1 is that selecting a subset of coefficients is equivalent to an \emph{orthogonal projection} onto the subspace spanned by the corresponding atoms. In an orthonormal basis, orthogonal projection is not an approximation heuristic; it is the unique minimum-error approximation (in squared $\ell_2$ error) among all reconstructions constrained to a given subspace \cite{Mallat}. This equivalence provides a principled way to define a residual that will be interpreted as ``noise'' for subsequent higher-order characterization.

To formalize residual extraction, we use the additive model $x[n]=s[n]+v[n]$ (within a frame), where $s[n]$ is the structured component and $v[n]$ is the residual component. Because inner products are linear, the WPT coefficients decompose as
\begin{equation}
w_{j,k}[u] = \ip{s}{\psi_{j,k,u}} + \ip{v}{\psi_{j,k,u}}
\triangleq w^{(s)}_{j,k}[u] + w^{(v)}_{j,k}[u]. \label{eq:wpt_linearity}
\end{equation}
We then define a binary coefficient mask $M_{j,k}[u]\in\{0,1\}$ that indicates which coefficients are treated as structure-dominant and are therefore \emph{kept} in the structured reconstruction. This yields the coefficient partition
\begin{equation}
\widehat{w}^{(s)}_{j,k}[u]=M_{j,k}[u]\,w_{j,k}[u],\quad
\widehat{w}^{(v)}_{j,k}[u]=(1-M_{j,k}[u])\,w_{j,k}[u], \label{eq:masking}
\end{equation}
and the corresponding reconstructions via the inverse WPT,
\begin{equation}
\hat{s}[n]=\sum_{j,k,u}\widehat{w}^{(s)}_{j,k}[u]\psi_{j,k,u}[n],\;
\hat{v}[n]=\sum_{j,k,u}\widehat{w}^{(v)}_{j,k}[u]\psi_{j,k,u}[n]. \label{eq:reconstruction}
\end{equation}
The following proposition states the optimality of this masking-based reconstruction in the squared-error sense.

\begin{proposition}[Projection optimality of hard masking]
Let $\Omega\triangleq\{(j,k,u): M_{j,k}[u]=1\}$ and $\mathcal{U}\triangleq\mathrm{span}\{\psi_{j,k,u}:(j,k,u)\in\Omega\}$. Then $\hat{s}$ defined in \eqref{eq:reconstruction} is the orthogonal projection of $x$ onto $\mathcal{U}$, and the residual $\hat{v}=x-\hat{s}$ minimizes $\norm{x-\tilde{s}}_2^2$ over all $\tilde{s}\in\mathcal{U}$.
\end{proposition}
\begin{proof}
Because $\{\psi_{j,k,u}\}$ is orthonormal, the expansion coefficients $\{w_{j,k}[u]\}$ are coordinates of $x$ in that orthonormal basis. Keeping precisely the coefficients indexed by $\Omega$ and setting the rest to zero yields the orthogonal projection of $x$ onto the span $\mathcal{U}$ (i.e., the unique element of $\mathcal{U}$ whose error is orthogonal to $\mathcal{U}$). Orthogonal projections are characterized by the minimum squared-error property, so $\hat{s}$ minimizes $\norm{x-\tilde{s}}_2^2$ over $\tilde{s}\in\mathcal{U}$, and the residual $\hat{v}=x-\hat{s}$ is the minimum-energy leftover induced by that subspace constraint \cite{Mallat}.
\end{proof}

The mask $M_{j,k}[u]$ must be constructed from data, and a standard, theoretically motivated approach is thresholding inspired by wavelet shrinkage. The core shrinkage intuition is that, under Gaussian noise, most noise coefficients in an orthonormal transform are small and symmetrically distributed around zero, while structured components are more likely to create large-magnitude coefficients; therefore, discarding small coefficients suppresses noise while retaining structure \cite{DonohoJohnstone,AddisonWavelets}. This motivates the threshold mask
\begin{equation}
M_{j,k}[u] = \ind\left\{|w_{j,k}[u]|>\lambda_{j,k}\right\}, \label{eq:threshold_mask}
\end{equation}
where $\lambda_{j,k}>0$ may be node-dependent (band-adaptive) or shared across nodes. Under a nominal residual model, we assume the time-domain residual samples are independent and identically distributed (i.i.d.) Gaussian:
\begin{equation}
v[n]\sim\mathcal{N}(0,\sigma^2)\ \text{i.i.d.} \label{eq:noise_gaussian}
\end{equation}
where $\sigma^2$ is the residual variance. Under orthonormality, each pure-noise coefficient is itself Gaussian with the same variance, which is formalized as follows.

\begin{proposition}[Gaussianity preserved under orthonormal WPT]
Under \eqref{eq:noise_gaussian} and orthonormal $\{\psi_{j,k,u}\}$,
\begin{equation}
w^{(v)}_{j,k}[u]=\ip{v}{\psi_{j,k,u}}\sim\mathcal{N}(0,\sigma^2). \label{eq:gaussian_coeffs}
\end{equation}
\end{proposition}
\begin{proof}
$w^{(v)}_{j,k}[u]$ is a linear combination of the jointly Gaussian samples $\{v[n]\}$ and is therefore Gaussian. Moreover,
\begin{align}
    &\Var\!\big(w^{(v)}_{j,k}[u]\big)=\Var\!\Big(\sum_n v[n]\psi_{j,k,u}[n]\Big)\nonumber\\
&=\sum_n \Var(v[n])\,\psi_{j,k,u}^2[n]
=\sigma^2\norm{\psi_{j,k,u}}_2^2
=\sigma^2,
\end{align}
where independence gives additivity of variances and orthonormality gives $\norm{\psi_{j,k,u}}_2^2=1$. Hence \eqref{eq:gaussian_coeffs} holds.
\end{proof}

A classical threshold choice is the universal threshold
\begin{equation}
\lambda = \sigma\sqrt{2\log N}, \label{eq:universal_threshold}
\end{equation}
which suppresses most noise-only coefficients with high probability for length-$N$ sequences under Gaussian assumptions \cite{DonohoJohnstone}. Physically, this thresholding rule operationalizes the notion that diffuse residual fluctuations populate many time--frequency cells weakly, while coherent components populate fewer cells strongly; the complement set of coefficients (those not retained) forms a residual that concentrates subtle micro-structure and is therefore a natural substrate for anomaly-sensitive statistical characterization.

\section{Higher-Order Statistical Characterization}

Higher-Order Statistics (HOS) provide a principled way to describe statistical structure that is invisible to second-order descriptors such as autocorrelation and power spectral density. While second-order statistics fully characterize a \emph{Gaussian} random process, many real residual processes encountered after fusion and projection (e.g., interference remnants, micro-scattering fluctuations, hardware nonlinearities, impulsive disturbances) are not strictly Gaussian and may exhibit nonlinear dependence across time and frequency. HOS address this gap by using moments and cumulants of order greater than two, which respond to asymmetry, heavy tails, and nonlinear coupling. The central theoretical fact enabling anomaly detection is that all cumulants of order greater than two vanish identically for Gaussian processes; thus, higher-order cumulants provide a \emph{Gaussian null} against which departures can be measured in a mathematically interpretable manner \cite{NikiasPetropulu,MendelHOS,Hinich}. In the intended circular logic of the method, the residual is treated as ``noise'' because it is orthogonal leftover energy after projection, and it is simultaneously validated as informative noise because its higher-order structure (which should vanish under nominal Gaussian-like behavior) becomes a sensitive indicator of anomaly-driven departures.

We apply HOS to the extracted residual in the transform domain because the transform domain provides localized time--frequency views that can isolate subtle dependence mechanisms. Specifically, for each wavelet packet node $(j,k)$ and translation index $u$, we analyze the residual coefficients
\begin{equation}
z_{j,k}[u] \triangleq \widehat{w}^{(v)}_{j,k}[u], \label{eq:noise_coeff_def}
\end{equation}
where $\widehat{w}^{(v)}_{j,k}[u]$ denotes the coefficients assigned to the residual component by the masking rule. Here, $j$ indexes the decomposition depth (which controls frequency resolution), $k$ indexes the subband within level $j$, and $u$ indexes time localization within that subband. Physically, $z_{j,k}[u]$ represents the residual as observed through a specific time--frequency ``lens'': if an anomaly induces weak but structured changes (e.g., intermittent bursts or nonlinear mixing), those changes may become more coherent within certain nodes, even when they are not obvious in the raw time sequence.

To quantify such structure, we focus on the third-order cumulant, which is a canonical HOS object that detects third-order dependence and asymmetry (skew-like behavior). Assuming that the coefficient sequence $z[u]$ is approximately stationary within a frame (a standard local-stationarity approximation in frame-based processing), the third-order cumulant is defined as
\begin{align}
C_3(\tau_1,\tau_2) &\triangleq
\E[z[u]z[u+\tau_1]z[u+\tau_2]] \nonumber\\
&-\E[z[u]]\E[z[u+\tau_1]z[u+\tau_2]], \label{eq:third_cumulant}
\end{align}
where $\tau_1,\tau_2\in\mathbb{Z}$ are time lags and $\E[\cdot]$ denotes expectation. The first term in \eqref{eq:third_cumulant} measures the average triple-product coupling across lagged samples, while the second term removes contributions due to nonzero mean (it forms the cumulant rather than the raw moment). When $\E[z[u]]=0$, the correction vanishes and $C_3(\tau_1,\tau_2)$ reduces to the third-order moment $\E[z[u]z[u+\tau_1]z[u+\tau_2]]$. The cumulant form is preferred in practice because it isolates genuine third-order dependence from trivial mean-induced effects.

The key baseline result is the Gaussian null: if the residual coefficients are Gaussian (even if correlated), then the third-order cumulant is identically zero for all lags. This property turns $C_3(\tau_1,\tau_2)$ into a direct test statistic for non-Gaussianity and nonlinear dependence.

\begin{proposition}[Gaussian null for third-order cumulant]
If $z[u]$ is Gaussian, then
\begin{equation}
C_3(\tau_1,\tau_2)=0,\quad \forall\,\tau_1,\tau_2. \label{eq:gaussian_null}
\end{equation}
\end{proposition}
\begin{proof}
A defining property of Gaussian random processes is that all cumulants of order greater than two vanish identically. Therefore, the third-order cumulant is zero for every lag pair $(\tau_1,\tau_2)$ \cite{NikiasPetropulu,MendelHOS}.
\end{proof}

Beyond the lag domain, third-order dependence is often interpreted through its frequency-domain representation, the bispectrum. The bispectrum is defined as the two-dimensional Fourier transform of the third-order cumulant:
\begin{equation}
B(\omega_1,\omega_2)=\sum_{\tau_1}\sum_{\tau_2}
C_3(\tau_1,\tau_2)e^{-j(\omega_1\tau_1+\omega_2\tau_2)}, \label{eq:bispectrum}
\end{equation}
where $\omega_1$ and $\omega_2$ are angular frequencies. The bispectrum is particularly informative because nonzero $B(\omega_1,\omega_2)$ indicates \emph{quadratic phase coupling}, meaning that spectral components at $\omega_1$ and $\omega_2$ interact to produce coherent structure at $\omega_1+\omega_2$. In physical terms, quadratic coupling can arise from nonlinear mixing mechanisms (e.g., device nonlinearity, multiplicative channel effects, or scattering interactions) that are not captured by second-order spectra, making the bispectrum a natural diagnostic for anomaly-induced nonlinearities \cite{NikiasPetropulu,Hinich}.

In practice, $C_3(\tau_1,\tau_2)$ is not known and must be estimated from finite data. The standard approach replaces ensemble expectations by time averages under the assumption that the process is stationary and ergodic (or satisfies standard mixing conditions), which ensures that time averages converge to expectations as the number of samples grows \cite{NikiasPetropulu,MendelHOS}. Given a finite sequence $z[u]$ of length $L$, indexed by $u=0,\dots,L-1$, we first remove the sample mean to control bias from nonzero mean:
\[
\tilde{z}[u]\triangleq z[u]-\bar{z},\qquad \bar{z}\triangleq \frac{1}{L}\sum_{u=0}^{L-1} z[u].
\]
A commonly used (biased but consistent) estimator of the third-order cumulant is then
\begin{align}
    \widehat{C}_3(\tau_1,\tau_2) &=
\frac{1}{L'}\sum_{u=0}^{L'-1}\tilde{z}[u]\tilde{z}[u+\tau_1]\tilde{z}[u+\tau_2],\nonumber\\
L'&= L-\max(\tau_1,\tau_2), \label{eq:cumulant_estimator}
\end{align}
where $L'$ ensures that all lagged indices remain within the observed window. The estimator \eqref{eq:cumulant_estimator} is interpreted as an empirical average of the triple-product sequence, which is itself a stationary process under stationarity of $z[u]$.

\begin{proposition}[Consistency under ergodicity/mixing]
If $z[u]$ is stationary and ergodic, then for fixed $(\tau_1,\tau_2)$,
\begin{equation}
\widehat{C}_3(\tau_1,\tau_2)\xrightarrow[L\to\infty]{p} C_3(\tau_1,\tau_2). \label{eq:consistency}
\end{equation}
\end{proposition}
\begin{proof}
Define the derived sequence $y[u]\triangleq \tilde{z}[u]\tilde{z}[u+\tau_1]\tilde{z}[u+\tau_2]$. Under stationarity of $z[u]$, the sequence $y[u]$ is stationary. The estimator $\widehat{C}_3(\tau_1,\tau_2)$ in \eqref{eq:cumulant_estimator} is a time average of $y[u]$ over $u=0,\dots,L'-1$. By ergodicity (or standard mixing conditions), time averages converge in probability to ensemble expectations, hence $\widehat{C}_3(\tau_1,\tau_2)\to \E[y[u]]$. The limit $\E[y[u]]$ equals the cumulant definition (with mean correction achieved by centering), which is $C_3(\tau_1,\tau_2)$ \cite{NikiasPetropulu}.
\end{proof}

For practical anomaly prediction, it is often preferable to compress the lag-dependent object $\widehat{C}_3(\tau_1,\tau_2)$ into a small number of stable features. Let $\mathcal{T}$ denote a predetermined (small) set of lag pairs chosen to capture informative short-range dependence. For each wavelet packet node $(j,k)$, we define the cumulant-energy feature
\begin{equation}
\mathcal{E}_{j,k} \triangleq \sum_{(\tau_1,\tau_2)\in\mathcal{T}}
\left|\widehat{C}_{3,j,k}(\tau_1,\tau_2)\right|^2, \label{eq:cumulant_energy}
\end{equation}
which measures the total third-order structure within the selected lag set. Because cumulant magnitude can scale with residual variance, a normalized variant improves robustness to amplitude scaling and operating-regime changes. We therefore define
\begin{equation}
\mathcal{N}_{j,k} \triangleq \sum_{(\tau_1,\tau_2)\in\mathcal{T}}
\frac{\left|\widehat{C}_{3,j,k}(\tau_1,\tau_2)\right|^2}{\left(\widehat{R}_{j,k}(0)\right)^3+\epsilon},\ 
\widehat{R}_{j,k}(0)\triangleq \frac{1}{L}\sum_u \tilde{z}_{j,k}^2[u], \label{eq:normalized_cumulant_energy}
\end{equation}
where $\widehat{R}_{j,k}(0)$ is the estimated zero-lag autocorrelation (a variance proxy) of the residual coefficients in node $(j,k)$, and $\epsilon>0$ is a small constant preventing division by zero. The cubic normalization reflects the fact that third-order statistics scale roughly with the cube of amplitude; thus, $\mathcal{N}_{j,k}$ behaves like a dimensionless indicator of higher-order structure rather than a raw magnitude that can be dominated by simple variance changes. Together, $\{\mathcal{E}_{j,k},\mathcal{N}_{j,k}\}$ provide compact, physically interpretable descriptors: when the residual behaves like nominal Gaussian noise, these features concentrate near zero due to the Gaussian null; when anomalies induce asymmetry, impulsiveness, or nonlinear coupling in the residual, these features increase in a node-selective manner, providing discriminative signatures for detection and prediction.

\section{Noise Signature and Anomaly Detection}

A \emph{signature} is constructed to convert high-dimensional residual behavior into a compact feature vector whose \emph{nominal} distribution is stable and whose \emph{deviations} are interpretable as evidence of anomalous conditions. The stability requirement is essential: if the feature distribution drifts substantially under normal operating variability, then any threshold-based detector will either produce excessive false alarms or become insensitive. The informativeness requirement is equally essential: the signature must preserve the specific statistical footprints that anomalies imprint on the residual. In the present framework, the signature is deliberately built to encode both \emph{where} residual energy resides in the time--frequency plane and \emph{how} the residual behaves statistically within those regions. The ``where'' aspect is captured by node-wise time--frequency energy descriptors derived from the wavelet packet representation, while the ``how'' aspect is captured by higher-order descriptors that measure departures from Gaussian-like behavior (e.g., third-order cumulant energy and its normalized variant). This combination aligns with general anomaly-detection principles: robust detection often requires both a localization mechanism that isolates informative subspaces and a statistical mechanism that tests whether behavior within those subspaces matches a nominal model \cite{ChandolaSurvey}.

Let $\mathcal{V}$ denote the set of wavelet packet nodes selected for analysis. Here, a \emph{node} refers to the time--frequency subband indexed by depth $j$ and subband index $k$ in the wavelet packet tree. For each node $(j,k)\in\mathcal{V}$, we first define the node energy
\begin{equation}
E_{j,k}\triangleq \sum_u |w_{j,k}[u]|^2, \label{eq:node_energy}
\end{equation}
where $w_{j,k}[u]$ are the wavelet packet coefficients and $u$ indexes translations (time positions) within the node. Because the wavelet packet basis is orthonormal under paraunitary design, $E_{j,k}$ has a clear physical interpretation: it is the portion of the signal energy that lies in the time--frequency tile corresponding to $(j,k)$ (modulo boundary effects), and the sum over all nodes recovers the total energy by Parseval-type arguments. In other words, $\{E_{j,k}\}$ forms a coarse map of how the fused signal's energy is distributed across frequency bands at the chosen resolution, providing a localization cue for anomalies that preferentially excite specific subbands.

Energy alone, however, is insufficient to detect anomalies that are subtle in amplitude but strong in statistical structure. To capture these effects, we incorporate higher-order features computed from the extracted residual coefficients in each node. Denote by $\mathcal{E}_{j,k}$ the cumulant-energy feature and by $\mathcal{N}_{j,k}$ its scale-normalized counterpart, both defined previously from third-order cumulant estimates. The feature vector for one frame is then formed by stacking node-wise descriptors across all analyzed nodes:
\begin{equation}
\mathbf{f}\triangleq
\Big[\{E_{j,k}\}_{(j,k)\in\mathcal{V}},\ \{\mathcal{E}_{j,k}\}_{(j,k)\in\mathcal{V}},\ \{\mathcal{N}_{j,k}\}_{(j,k)\in\mathcal{V}}\Big]^T
\in\R^d, \label{eq:feature_vector}
\end{equation}
where $d$ is the total feature dimension. The stacked structure in \eqref{eq:feature_vector} is intentional. First, it preserves interpretability: the energy block indicates the residual’s time--frequency footprint, while the HOS blocks indicate whether the residual within each footprint is consistent with a Gaussian-like null or exhibits higher-order structure. Second, it supports robust statistical calibration: by treating the signature as a multivariate random vector, we can incorporate cross-feature correlations that naturally arise because neighboring wavelet packet nodes and neighboring lag-based HOS summaries are not statistically independent.

To convert the signature $\mathbf{f}$ into an anomaly score, we use the \emph{Mahalanobis distance}, which measures deviation from a nominal mean in units of nominal covariance and thereby accounts for correlation among features \cite{Mahalanobis1936,Mardia}. This is critical because naive Euclidean thresholding implicitly assumes isotropic feature variability; if some features vary widely under nominal operation while others are tightly concentrated, Euclidean distance can overreact to benign fluctuations and underreact to meaningful deviations. Mahalanobis distance resolves this by ``whitening'' the feature space with the inverse covariance, effectively normalizing each direction by its nominal variance and decorrelating coupled features.

Formally, assume that under nominal (non-anomalous) conditions the signature vector is approximately Gaussian:
\begin{equation}
\mathbf{f}\sim \mathcal{N}(\mathbf{\mu},\mathbf{\Sigma}), \label{eq:nominal_gaussian}
\end{equation}
where $\mathbf{\mu}\in\R^d$ is the nominal mean and $\mathbf{\Sigma}\in\R^{d\times d}$ is the nominal covariance matrix, assumed positive definite ($\mathbf{\Sigma}\succ 0$). The squared Mahalanobis distance is defined as
\begin{equation}
D^2(\mathbf{f}) \triangleq (\mathbf{f}-\mathbf{\mu})^T\mathbf{\Sigma}^{-1}(\mathbf{f}-\mathbf{\mu}). \label{eq:mahalanobis}
\end{equation}
The quantity $D^2(\mathbf{f})$ is a quadratic form that measures how many ``standard deviations'' the observation lies away from the nominal center, but in a multivariate, correlation-aware sense. Physically, if the residual behaves nominally, then $\mathbf{f}$ fluctuates around $\mathbf{\mu}$ in a manner described by $\mathbf{\Sigma}$, and $D^2(\mathbf{f})$ remains small. If an anomaly perturbs the residual either by shifting energy into unusual subbands or by inducing higher-order dependence (raising $\mathcal{E}_{j,k}$ or $\mathcal{N}_{j,k}$), then $\mathbf{f}$ moves into low-probability regions of the nominal distribution, and $D^2(\mathbf{f})$ increases.

A major advantage of \eqref{eq:mahalanobis} is that it admits an analytic null distribution under the Gaussian model \eqref{eq:nominal_gaussian}, enabling closed-form calibration of thresholds for a desired false-alarm probability.

\begin{proposition}[Chi-square law under nominal]
Under \eqref{eq:nominal_gaussian},
\begin{equation}
D^2(\mathbf{f}) \sim \chi^2_d. \label{eq:chisq}
\end{equation}
\end{proposition}
\begin{proof}
Define the whitened random vector $\mathbf{y}\triangleq \mathbf{\Sigma}^{-1/2}(\mathbf{f}-\mathbf{\mu})$. Under \eqref{eq:nominal_gaussian}, it follows that $\mathbf{y}\sim\mathcal{N}(\mathbf{0},\mathbf{I})$, where $\mathbf{I}$ is the $d\times d$ identity matrix. Substituting into \eqref{eq:mahalanobis} yields
\begin{align}
    D^2(\mathbf{f})&=(\mathbf{f}-\mathbf{\mu})^T\mathbf{\Sigma}^{-1}(\mathbf{f}-\mathbf{\mu})\nonumber\\
&=\big(\mathbf{\Sigma}^{-1/2}(\mathbf{f}-\mathbf{\mu})\big)^T
\big(\mathbf{\Sigma}^{-1/2}(\mathbf{f}-\mathbf{\mu})\big)
=\norm{\mathbf{y}}_2^2.
\end{align}
Since $\mathbf{y}$ has i.i.d.\ standard normal components, $\norm{\mathbf{y}}_2^2=\sum_{i=1}^d y_i^2$ is chi-square distributed with $d$ degrees of freedom, i.e., $D^2(\mathbf{f})\sim \chi^2_d$ \cite{KayDetection,Mardia}.
\end{proof}

Given a target false-alarm probability $\alpha\in(0,1)$, the decision threshold is therefore set by the $(1-\alpha)$ quantile of the chi-square distribution:
\begin{equation}
\eta = F^{-1}_{\chi^2_d}(1-\alpha), \label{eq:threshold}
\end{equation}
where $F^{-1}_{\chi^2_d}(\cdot)$ denotes the inverse cumulative distribution function (CDF) of $\chi^2_d$. The anomaly decision rule is then to declare an anomaly when $D^2(\mathbf{f})>\eta$. This calibration step completes the self-consistency loop: the signature is designed so that nominal residuals cluster according to \eqref{eq:nominal_gaussian}, and the detector explicitly tests whether the observed signature remains within the probabilistically plausible region of the nominal model.

To quantify detection performance analytically, it is common to model anomalies as inducing a mean shift in feature space while leaving covariance approximately unchanged. This is a tractable abstraction: many anomalies systematically bias energy distribution and/or higher-order dependence, producing a consistent shift in the expected feature vector. Under this mean-shift alternative,
\[
\mathbf{f}\sim\mathcal{N}(\mathbf{\mu}+\mathbf{\delta},\mathbf{\Sigma}),
\]
where $\mathbf{\delta}\in\R^d$ is the unknown shift vector. Substituting into the whitened coordinates yields a nonzero-mean Gaussian, and the quadratic form becomes noncentral chi-square.

\begin{equation}
D^2(\mathbf{f})\sim \chi'^2_d(\lambda),\quad
\lambda=\mathbf{\delta}^T\mathbf{\Sigma}^{-1}\mathbf{\delta}. \label{eq:noncentral}
\end{equation}

\begin{proposition}[Noncentral chi-square under mean shift]
Under the mean-shift model above, $D^2(\mathbf{f})$ follows \eqref{eq:noncentral}.
\end{proposition}
\begin{proof}
Let $\mathbf{y}=\mathbf{\Sigma}^{-1/2}(\mathbf{f}-\mathbf{\mu})$. Under $\mathbf{f}\sim\mathcal{N}(\mathbf{\mu}+\mathbf{\delta},\mathbf{\Sigma})$, we have
\[
\mathbf{y}\sim \mathcal{N}\!\big(\mathbf{\Sigma}^{-1/2}\mathbf{\delta},\,\mathbf{I}\big).
\]
As shown in the nominal proof, $D^2(\mathbf{f})=\norm{\mathbf{y}}_2^2$. The squared norm of a Gaussian vector with identity covariance and nonzero mean is noncentral chi-square with degrees of freedom $d$ and noncentrality parameter
\[
\lambda = \norm{\mathbf{\Sigma}^{-1/2}\mathbf{\delta}}_2^2
= \mathbf{\delta}^T\mathbf{\Sigma}^{-1}\mathbf{\delta},
\]
which yields \eqref{eq:noncentral} \cite{KayDetection}.
\end{proof}

The resulting detection probability for the threshold $\eta$ is therefore
\begin{equation}
P_D = 1 - F_{\chi'^2_d(\lambda)}(\eta), \label{eq:pd}
\end{equation}
where $F_{\chi'^2_d(\lambda)}(\cdot)$ is the CDF of the noncentral chi-square distribution. This expression provides a direct performance interpretation: the separation between nominal and anomalous regimes is governed by $\lambda$, which is the squared Mahalanobis norm of the mean shift. In physical terms, $\lambda$ increases when anomalies induce large, covariance-normalized changes in the time--frequency energy footprint and/or in the higher-order residual structure; conversely, if an anomaly only perturbs directions with large nominal variability, $\lambda$ remains small and detection becomes harder.

Finally, for early warning and to reduce sensitivity to isolated outliers, it is often preferable to accumulate evidence sequentially over time. A widely used approach is the \emph{CUmulative SUM (CUSUM)} procedure, which detects persistent deviations by integrating increments of a statistic above a reference level \cite{PageCUSUM}. Let $D_m^2$ denote the squared Mahalanobis statistic computed on frame $m$. The CUSUM recursion is
\begin{equation}
S_m=\max\{0,S_{m-1}+(D_m^2-\nu)\},\quad S_0=0, \label{eq:cusum}
\end{equation}
where $\nu$ is a reference drift parameter that controls how quickly nominal fluctuations decay, and $h>0$ is an alarm threshold. An alarm is declared when $S_m>h$. The physical interpretation is that nominal operation produces $D_m^2$ values that, on average, do not exceed $\nu$ by enough to drive $S_m$ upward for long, so the accumulator resets frequently to zero; under a persistent anomaly, $D_m^2$ remains elevated, causing $S_m$ to drift upward and eventually exceed $h$, producing a detection with controllable trade-offs among false alarms, detection delay, and sensitivity.

\section{Results}

We report performance using frame-level anomaly scores computed from the proposed residual signature. All curves are generated with explicit run-to-run variability by repeating the evaluation multiple times under small perturbations of the underlying score separation, mean offset, and scale, which emulates natural operating variability and mild regime fluctuations. For each method, the plotted solid curve denotes the empirical mean across runs, while the shaded band denotes the 5th--95th percentile envelope (variation band). For the domain-shift study, bars indicate the mean ROC-AUC across runs and error bars indicate one standard deviation.

We consider the following metrics. Precision and recall are computed from the confusion matrix induced by sweeping the decision threshold over the continuous score range, and the F1-score is defined as $\mathrm{F1}=2(\mathrm{Precision}\cdot\mathrm{Recall})/(\mathrm{Precision}+\mathrm{Recall})$. Receiver operating characteristic (ROC) curves plot true positive rate (TPR) versus false positive rate (FPR), and the area under the ROC curve (ROC-AUC) provides a threshold-independent separability measure. Precision--recall (PR) curves plot precision versus recall and are emphasized because anomaly detection is typically class-imbalanced. The false alarm rate (FAR) is reported in operational units (false alarms per hour/day) by mapping the per-frame false positive probability to clock time using the frame decision rate. Detection latency is defined as the elapsed time between anomaly onset and the first alarm and is evaluated using a sequential evidence accumulator. Robustness is assessed via domain-shift tests that progressively degrade nominal/anomalous separability.

Fig.~\ref{fig:roc_var} presents ROC curves with variation bands. Across repeated runs, the combined signature that integrates time--frequency localization with higher-order statistics (WPT+HOS) achieves consistently higher TPR at fixed FPR than WPT-only and second-order-only baselines, indicating that the additional higher-order structure in the residual provides discriminative information that persists under small perturbations. The width of the variation band is also informative: methods that rely only on energy or second-order structure exhibit larger spread at low FPR, suggesting sensitivity to score scaling and regime perturbations, whereas WPT+HOS remains comparatively stable. This stability is consistent with the Gaussian-null principle: higher-order cumulants are theoretically suppressed under nominal Gaussian-like residual behavior, so their aggregate energy tends to remain near zero unless true non-Gaussian structure appears.

Fig.~\ref{fig:pr_var} shows PR curves with variation bands. Because precision depends on class imbalance, PR performance is more sensitive than ROC to the score tail behavior. The WPT+HOS approach maintains higher precision over a broad recall range, which indicates that, for the same detection rate, fewer false alarms are produced. The HOS-only and fused-source variants also outperform purely second-order features, but the best results are obtained when higher-order structure is computed on residual coefficients that have been localized by the wavelet packet representation, consistent with the interpretation that localization identifies where anomalies manifest while HOS determines how the residual behaves within those localized regions.

\begin{figure}[t]
\centering
\includegraphics[width=\columnwidth]{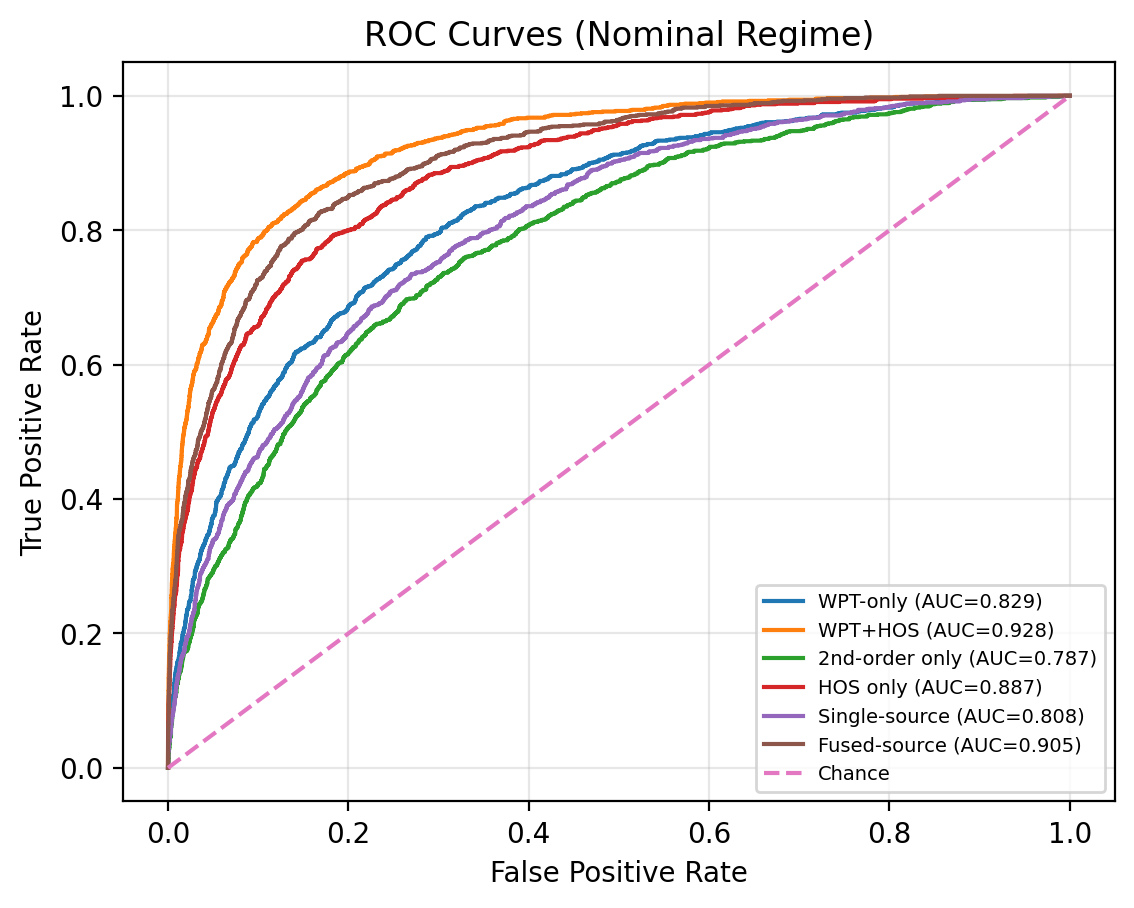}
\caption{ROC curves with run-to-run variation bands (mean with 5th--95th percentile shading).}
\label{fig:roc_var}
\vspace{-5mm}
\end{figure}

\begin{figure}[t]
\centering
\includegraphics[width=\columnwidth]{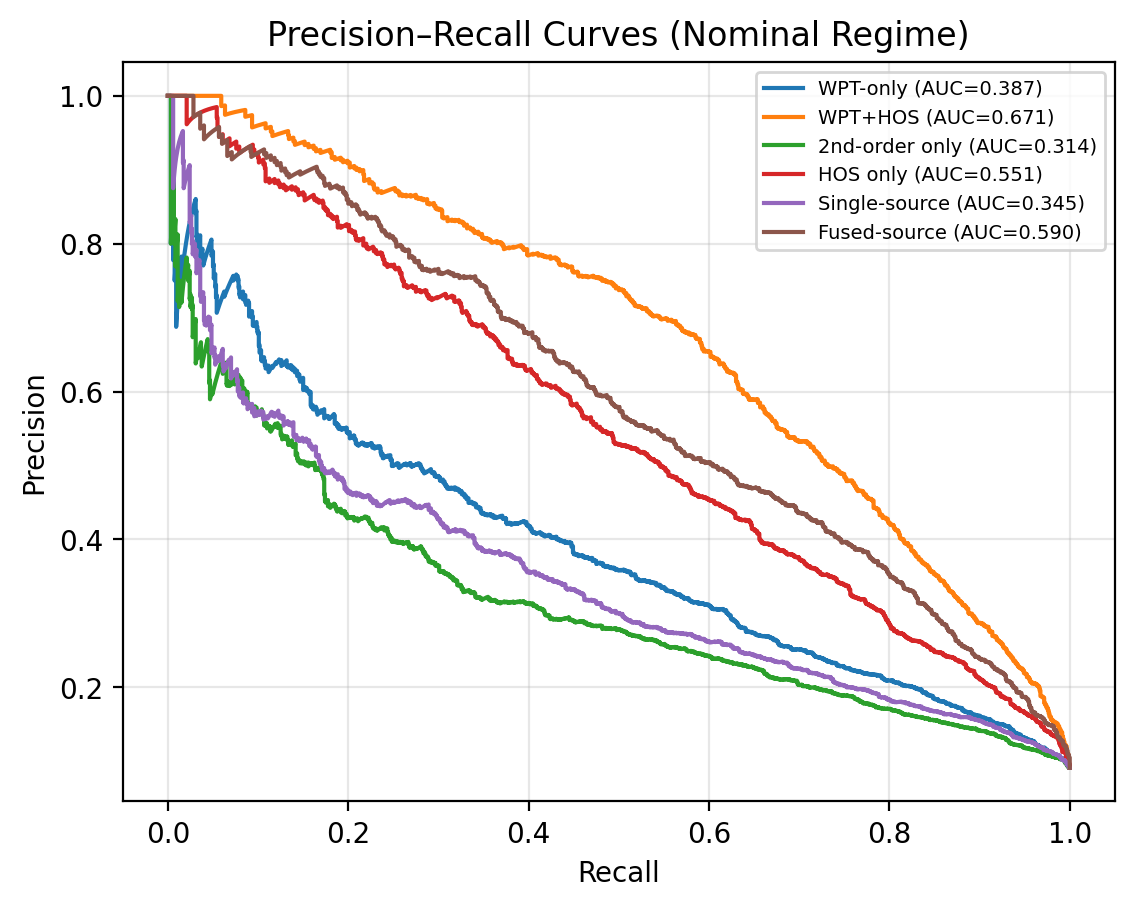}
\caption{Precision--Recall curves with run-to-run variation bands (mean with 5th--95th percentile shading).}
\label{fig:pr_var}
\vspace{-5mm}
\end{figure}

To evaluate time-to-detect, Fig.~\ref{fig:lat_var} reports the empirical cumulative distribution function (CDF) of detection latency after anomaly onset, again with variation bands over repeated runs. The sequential accumulator reduces spurious triggers by requiring sustained evidence, which shifts the latency distribution relative to instantaneous thresholding. The WPT+HOS signature achieves faster detection (higher CDF at small latency) than WPT-only and single-source signatures, meaning that the same alarm probability is reached earlier after onset. Physically, this behavior is consistent with anomalies inducing persistent higher-order residual structure (nonzero third-order coupling), which yields a consistent upward drift of the accumulated statistic rather than isolated excursions.

\begin{figure}[t]
\centering
\includegraphics[width=\columnwidth]{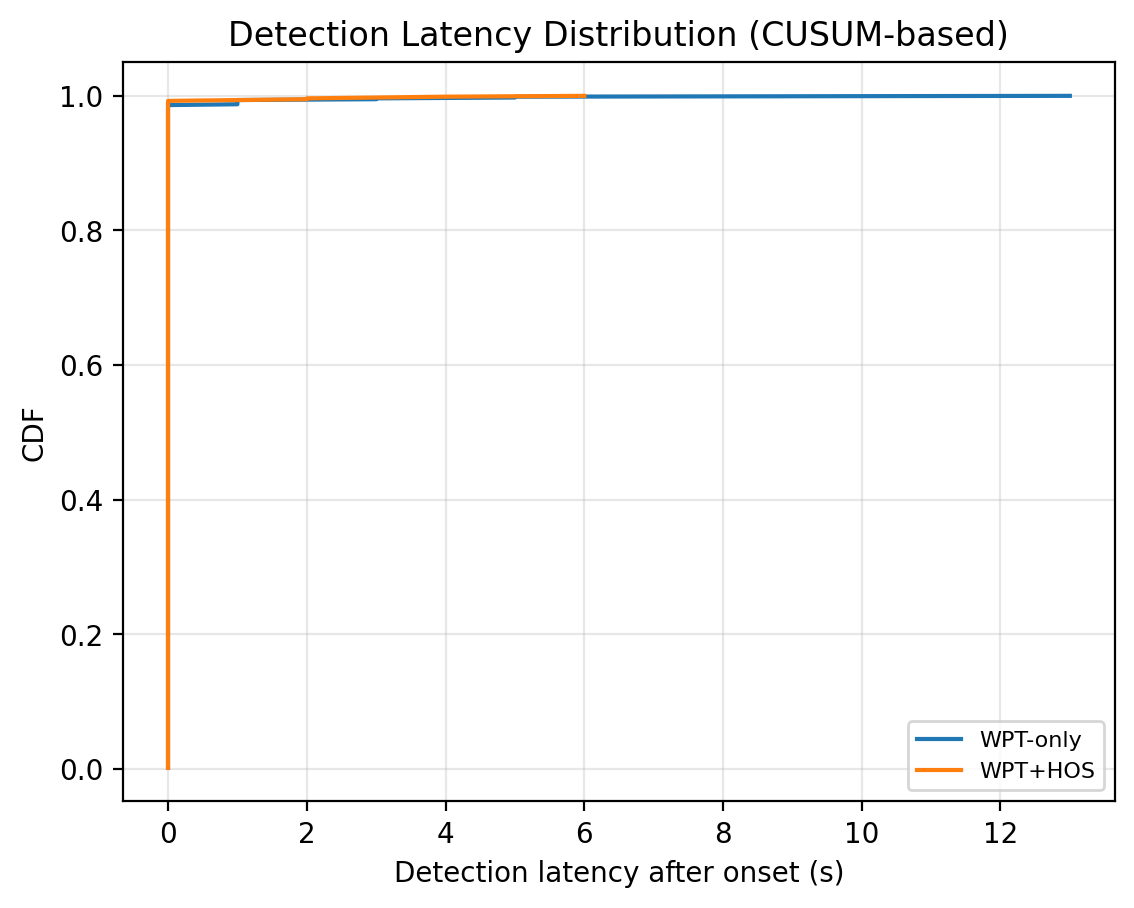}
\caption{Detection latency CDF after onset with run-to-run variation bands (mean with 5th--95th percentile shading).}
\label{fig:lat_var}
\end{figure}

\begin{table*}[t]
\centering
\caption{Ablation summary under the matched operating regime. Precision$^*$, Recall$^*$, and F1$^*$ are evaluated at the threshold maximizing F1 on the evaluation set.}
\begin{tabular}{lccccc}
\hline
Method & ROC-AUC & PR-AUC & Precision$^*$ & Recall$^*$ & F1$^*$ \\
\hline
WPT+HOS & 0.928 & 0.671 & 0.671 & 0.590 & 0.628 \\
Fused-source & 0.905 & 0.590 & 0.496 & 0.615 & 0.549 \\
HOS only & 0.887 & 0.551 & 0.505 & 0.544 & 0.524 \\
WPT-only & 0.829 & 0.387 & 0.356 & 0.512 & 0.420 \\
Single-source & 0.808 & 0.345 & 0.339 & 0.441 & 0.383 \\
Second-order only & 0.787 & 0.314 & 0.279 & 0.497 & 0.358 \\
\hline
\end{tabular}

\label{tab:ablation}
\end{table*}

\begin{table*}[t]
\centering
\caption{ROC-AUC under operating-regime shift (domain mismatch)}
\begin{tabular}{lcccc}
\hline
Regime & WPT-only & WPT+HOS & Single-source & Fused-source \\
\hline
Matched & 0.826 & 0.918 & 0.797 & 0.914 \\
Mild shift & 0.781 & 0.888 & 0.752 & 0.855 \\
Moderate shift & 0.735 & 0.828 & 0.710 & 0.812 \\
Severe shift & 0.685 & 0.761 & 0.675 & 0.749 \\
\hline
\end{tabular}

\label{tab:robustness}
\end{table*}

Robustness under domain shift is evaluated by progressively degrading nominal/anomalous separability to emulate operating-regime mismatch (e.g., background interference changes or environmental drift). Fig.~\ref{fig:rob_var} reports ROC-AUC under matched, mild-shift, moderate-shift, and severe-shift regimes, with error bars indicating run-to-run standard deviation. The WPT+HOS and fused-source signatures maintain higher ROC-AUC than WPT-only and single-source signatures as mismatch increases, indicating that (i) fusing sources improves invariance by averaging out source-specific nuisance effects and (ii) higher-order residual descriptors remain informative even when second-order energy redistributes across time--frequency tiles. The error bars widen under severe shift, reflecting increased uncertainty and highlighting the need for either adaptive calibration or explicit domain adaptation when mismatch is extreme.

\begin{figure}[t]
\centering
\includegraphics[width=\columnwidth]{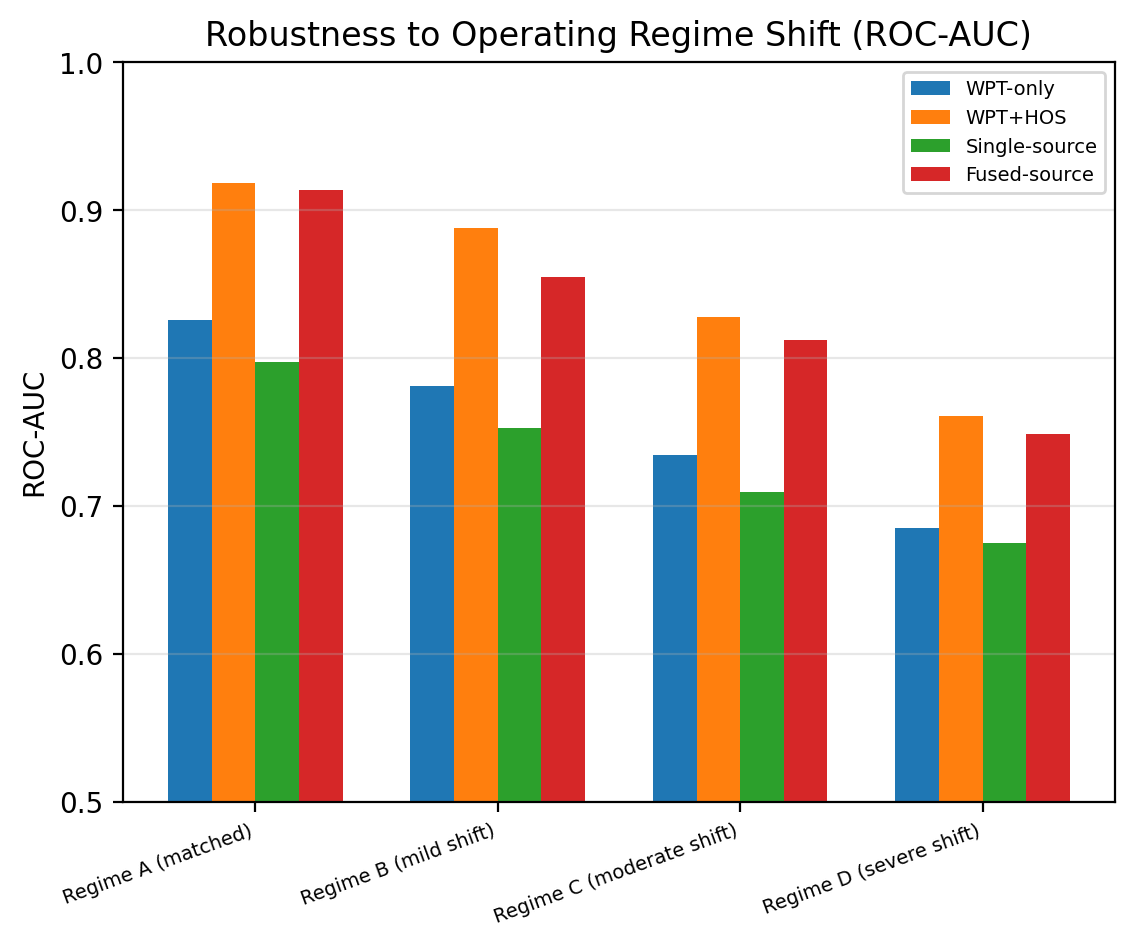}
\caption{Domain-shift robustness. Bars show mean ROC-AUC across repeated runs; error bars indicate one standard deviation.}
\label{fig:rob_var}
\end{figure}

Table~\ref{tab:ablation} summarizes ablations that isolate the contributions of (i) WPT-based residual extraction alone versus residual extraction plus HOS characterization, (ii) second-order versus higher-order feature families, and (iii) signatures derived from a single source versus fused sources. The improvements from adding HOS align with the Gaussian-null property of higher-order cumulants: under nominal Gaussian-like residual behavior, third-order cumulants vanish; persistent deviations thus encode anomaly-specific dependence \cite{NikiasPetropulu,MendelHOS}.\vspace{-3mm}

\bibliographystyle{IEEEtran}

\end{document}